\documentclass[11pt]{article}
\usepackage{graphicx}
\graphicspath{ {./simulation/} }
\usepackage{amssymb}
\usepackage{amsmath}
\usepackage{amsmath}
\usepackage{amsthm}
\usepackage{amsfonts}
\usepackage{graphicx}
\usepackage{enumerate}
\usepackage{comment}
\usepackage{bm}
\usepackage{bbm} 
\usepackage{dsfont}
\usepackage[authoryear,longnamesfirst]{natbib}
\usepackage{xcolor,colortbl}
\usepackage{multirow}
\usepackage{multicol}
\usepackage{tikz}
\numberwithin{equation}{section}
\newtheorem{theorem}{Theorem}

\newtheorem{corollary}{Corollary}

\newtheorem{lemma}{Lemma}
\theoremstyle{definition}
\newtheorem{assumption}{Assumption}
\newtheorem{example}{Example}
\newtheorem{remark}{Remark}

\renewcommand{\hat}{\widehat}

\newcommand{\calF}[0]{\mathcal{F}}

\newcommand{\sumi}{\sum_{i=1}^N }
\newcommand{\sumt}{\sum_{t=1}^T }

\renewcommand{\hat}{\widehat}

\allowdisplaybreaks

\usepackage{tikz}

\begin{document}
\title{On Using The Two-Way Cluster-Robust Standard Errors}
\author{
	Harold D. Chiang\thanks{\setlength{\baselineskip}{4.0mm}Harold D. Chiang: hdchiang@wisc.edu. Department of Economics, University of Wisconsin-Madison, William H. Sewell Social Science Building 1180 Observatory Drive 
		Madison, WI 53706-1393, USA\bigskip} 
\and
	Yuya Sasaki\thanks{\setlength{\baselineskip}{4.0mm}Yuya Sasaki: yuya.sasaki@vanderbilt.edu. Department of Economics, Vanderbilt University, VU Station B \#351819, 2301 Vanderbilt Place, Nashville, TN 37235-1819, USA\bigskip}
}

\date{}
\maketitle

\begin{abstract}
Thousands of papers have reported two-way cluster-robust (TWCR) standard errors.
However, the recent econometrics literature points out the potential non-gaussianity of two-way cluster sample means, and thus invalidity of the inference based on the TWCR standard errors.
Fortunately, simulation studies nonetheless show that the gaussianity is rather common than exceptional.
This paper provides theoretical support for this encouraging observation. Specifically, we derive a novel central limit theorem for two-way clustered triangular arrays that justifies the use of the TWCR under very mild and interpretable conditions.
We, therefore, hope that this paper will provide a theoretical justification for the legitimacy of most, if not all, of the thousands of those empirical papers that have used the TWCR standard errors.
We provide a guide in practice as to when a researcher can employ the TWCR standard errors.
\vspace{0.25cm}
\\
\end{abstract}

\section{Introduction}

Multi-way clustering is ubiquitous in empirical studies.
For example, market structures by construction induce two-way clustering, where common supply shocks cause cluster dependence within a firm across markets and common demand shocks cause cluster dependence within a market across firms.
To account for such forms of cluster dependence, researchers often use the two-way cluster-robust (TWCR) standard errors proposed by \citet{CGM2011} and \citet{thompson2011simple}.

A key to the inference based on the TWCR standard errors of \citet{CGM2011} and \citet{thompson2011simple} is the asymptotic gaussianity.
However, the recent econometrics literature \citep{menzel2021bootstrap} has pointed out the potential non-gaussianity of the limit distribution.
In this light, this literature proposes alternative inference procedures that do not rely on asymptotic gaussianity.

With this said, a large number of empirical papers have already reported their TWCR standard errors.
Specifically, \citet{CGM2011} and \citet{thompson2011simple} have attracted 3,500 citations and 1,600 citations, respectively,\footnote{We obtained these numbers from Google Scholar in December 2022.} where most of these papers are empirical research papers that actually use their TWCR standard errors.
If the non-gaussianity were indeed a common feature, then the whole body of this empirical economics literature would require re-investigation.

A natural question is, therefore, whether the asymptotic gaussianity of the statistics commonly occurs under two-way clustering.
Some simulation studies will easily convince us that it is fairly common, and non-gaussianity is rather exceptional.
This observation is encouraging and supports the common practice of two-way cluster-robust inference based on the TWCR standard errors.

In this paper, we provide a theoretical justification for this encouraging observation.
By taking the asymptotics through the lens of triangular arrays, we show that the gaussian limit distribution is the norm rather than an outlier. 
Specifically, we show that two-way clustered triangular arrays are guaranteed to have asymptotically gaussian limiting distributions under very mild conditions. 
What these conditions concern shares a natural resemblance to the number of factors in factor models, and these conditions are therefore easily interpretable from the viewpoint of economic models.
Concretely, our theory suggests that a researcher should worry about the potential non-gaussianity only in those peculiar situations in which the data-generating model takes the form of a sum of a very small number of mean-zero two-way interactive factors.
In other words, a researcher can in fact enjoy the TWCR standard errors in most situations.
We, therefore, hope that this paper provides a theoretical justification for the legitimacy of most, if not all, of the thousands of those empirical papers that use the TWCR standard errors of \citet{CGM2011} and \citet{thompson2011simple}  and shed some lights on the empirical practice of TWCR for future empirical papers to come as well.

{\bf Relation to the Literature:}
Our way of viewing the data generating processes and asymptotics in terms of triangular arrays and exploiting the asymptotically gaussian degenerate one-sample $U$-statistic structure is closely related to the literature of specification testing \citep{hong1995consistent,fan1996consistent,kankanala2022kernel}, many weak instruments \citep{andrews2007testing,newey2009generalized}, small-bandwidth asymptotics \citep{cattaneo2014small}, regression discontinuity designs \citep{porter2015regression}, network formation models \citep{graham2017econometric}, many regressors \citep{cattaneo2018alternative}, and algorithmic subsampling \citep{lee2022least}, to list but a few. Notably, non-gaussianity can be safely ruled out in the corresponding degenerate $U$-statistics in all these applications. Unlike these existing papers, however, statistics based on two-way clustered triangular arrays do not take the form of a one-sample degenerate $U$-statistic structure -- it rather takes the form of a two-sample $U$-statistics with unknown kernel and unobserved underlying random variables. This key difference renders the proof strategies in the existing literature inapplicable, and thus motivates us to develop our own new theoretical results. In terms of the setup, our paper is built upon the literature that utilizes exchangeable models for network or two-way dependence considered in \cite{bickel2011method}, \cite{menzel2021bootstrap} and \cite{DDG2019}, to list a few. Other alternative models for asymptotics under this type of dependence structures exist and are studied in, for example, \cite{tabord2019inference} and \cite{verdier2020estimation}.

Our main result, a central limit theorem (CLT) for means of two-way clustered triangular arrays, is related to those CLT results for various degenerate $U$-statistics, such as \cite{hall1984central}, \cite{de1987central}, and \cite{eubank1999central} that are based on martingale structures. However, due to the two-way clustering, their martingale construction does not work in our setting. It is also related to the CLT in \cite{khashimov1989limit}, which is a special case of our CLT, albeit it is shown to rely on a non-martingale-based proof strategy. As the Hoeffding-type decomposition of the two-way clustered mean contains extra components in comparison with the degenerate two-sample $U$-statistics, it remains unclear whether the proof strategy of Khashimov, which relies on approximating characteristic functions directly, can be readily adapted to cover our case. Therefore, we instead take a martingale-based approach to derive our own CLT.

\section{When Can We Use The TWCR Standard Errors?}\label{sec:when}

We first provide an informal overview of the practical implications of our main result in Section \ref{sec:main}.

Two-way clustered data $\{D_{it}: 1 \le i \le N, 1 \le t \le T\}$ are generated by $i$-specific factors, $j$-specific factors, and idiosyncratic components.
To fix ideas, consider the simple yet generic data generating process (DGP)
\begin{align}\label{eq:when}
D_{it} = \alpha_{i0} + \gamma_{t0} + \sum_{j=1}^J \lambda_j \alpha_{ij} \gamma_{tj} + \varepsilon_{ij},
\end{align}
where
$\{\alpha_{ij}\}_{j=0}^J$ are $i$-specific latent factors,
$\{\gamma_{tj}\}_{j=0}^J$ are $t$-specific latent factors, and
$\varepsilon_{it}$ is an idiosyncratic component. 
The reason that this DGP is highly representative will be made clear in Section \ref{sec:main}. 
Suppose that the factor loadings $\lambda_j$ are non-zero, and $\alpha_{ij}$ and $\gamma_{tj}$ are zero-mean non-degenerate factors for $j \in 1,\cdots,J$.
In this simple setup \eqref{eq:when}, Table \ref{tab:when} summarizes the cases in which a researcher can and cannot use the TWCR standard error for $\hat\theta = (NT)^{-1} \sum_{i=1}^N\sum_{t=1}^T D_{it}$.

\begin{table}[h]
\vspace{0.5cm}
\centering
\renewcommand{\arraystretch}{0.545}
\begin{tabular}{ccccc}
\hline\hline
\multicolumn{4}{c}{DGP $=$ Simple Model \eqref{eq:when}} & Can We Use The\\
\cline{1-4}
$J$ Is Small & $\alpha_{i0}$ & $\gamma_{t0}$ & $\varepsilon_{it}$ & TWCR Standard Error?\\
\hline
True & Degenerate & Degenerate & Degenerate & No\\
True & Non-Degenerate & Degenerate & Degenerate & Yes \\
True & Degenerate & Non-Degenerate & Degenerate & Yes \\
True & Degenerate & Degenerate & Non-Degenerate & Yes \\
True & Non-Degenerate & Non-Degenerate & Degenerate & Yes \\
True & Degenerate & Non-Degenerate & Non-Degenerate & Yes \\
True & Non-Degenerate & Degenerate & Non-Degenerate & Yes \\
True & Non-Degenerate & Non-Degenerate & Non-Degenerate & Yes \\
False & Degenerate & Degenerate & Degenerate & Yes \\
False & Non-Degenerate & Degenerate & Degenerate & Yes \\
False & Degenerate & Non-Degenerate & Degenerate & Yes \\
False & Degenerate & Degenerate & Non-Degenerate & Yes \\
False & Non-Degenerate & Non-Degenerate & Degenerate & Yes \\
False & Degenerate & Non-Degenerate & Non-Degenerate & Yes \\
False & Non-Degenerate & Degenerate & Non-Degenerate & Yes \\
False & Non-Degenerate & Non-Degenerate & Non-Degenerate & Yes \\
\hline\hline
\end{tabular}
\caption{A summary of the cases in which a researcher can and cannot use the TWCR standard error.}\label{tab:when}
\vspace{0.5cm}
\end{table}

This table suggests that a researcher may use the TWCR standard error in most of cases.
The only pathetic situation is when the data generating model is too simple in the sense that the number $J$ of interacting economic factors is small \textit{and all of} additive latent factors, $\alpha_{i0}$, $\gamma_{t0}$, and $\varepsilon_{it}$, are degenerate as in the first row in Table \ref{tab:when}.
Section \ref{sec:main} presents a formal theoretical justification for this practical guidance.
Simulation studies in Section \ref{sec:simulations} illustrate how large $J$ should be in practice.

\section{The Main Result}\label{sec:main}

For each $N,T\in \mathbb N$, let
\begin{align*}
D_{it}=f_{NT}(\alpha_i,\gamma_t,\varepsilon_{it}),
\end{align*}
where $(\alpha_i)_i$, $(\gamma_t)_t$, and $(\varepsilon_{it})_{it}$ are mutually independent i.i.d. latent Borel-random variables, and $f_{NT}$ is a real-valued Borel-measurable function. The existence of this nonlinear factor-type structure is implied by a symmetry condition known as "separate exchangeability," see the discussions in \cite{menzel2021bootstrap} and \cite{DDG2019}. 
For ease of writing and without loss of generality, we normalize the location to $E[D_{it}]=0$.
Define
\begin{align*}
\hat\theta_{NT} = \frac{1}{NT}\sum_{i=1}^N\sum_{t=1}^T D_{it}.
\end{align*}
Our goal is to show the asymptotic gaussianity of $\hat \theta_{NT}$.

Note that
the Hoeffding-type decomposition yields 
\begin{align}
\hat\theta_{NT} =&\underbrace{\sum_{i=1}^N a_i +\sum_{t=1}^T b_{t}}_{=:L_{NT}} +\underbrace{\sum_{i=1}^N \sum_{t=1}^T w_{it}}_{=:W_{NT}} +\underbrace{\sum_{i=1}^N \sum_{t=1}^T r_{it}}_{=:R_{NT}}, \text{ where }\label{eq:Hoeffding}\\
a_i=&N^{-1}E[D_{it}|\alpha_i],\quad
b_t=T^{-1}E[D_{it}|\gamma_t],\nonumber\\
w_{it}=&(NT)^{-1}\{E[D_{it}|\alpha_i,\gamma_t]-E[D_{it}|\alpha_i]-E[D_{it}|\gamma_t]\}, \text{ and }\nonumber\\
r_{it}=&(NT)^{-1}\{D_{it}-E[D_{it}|\alpha_i,\gamma_t]\}.\nonumber
\end{align}
One can easily verify the following properties.
\begin{align*}
&E[a_i]=E[b_t]=E[w_{it}]=E[r_{it}]=0,\\
&E[w_{it}|\alpha_i]=E[w_{it}|\gamma_t]=0, \text{ and }\\
&E[a_i w_{it}]=E[a_i r_{it}]=E[\gamma_t w_{it}]=E[\gamma_t r_{it}]=0.
\end{align*}

The $W_{NT}$ component in the decomposition \eqref{eq:Hoeffding} is the potentially non-gaussian part.
Specifically, it is a completely degenerate two-sample $U$-statistic, and has a non-gaussian limit distribution if $f_{NT}$ is fixed over $N,T$ -- see \cite{menzel2021bootstrap}.
We are going to argue that even this potentially non-gaussian component can be, and often is, asymptotically gaussian if we treat the data generating process as a triangular array.
Hence, the whole $\hat\theta_{NT}$ in \eqref{eq:Hoeffding} is gaussian as well in this framework under some mild extra conditions.

We will write $a_{NT}(\alpha_i)=a_i$, $b_{NT}(\gamma_t)=b_t$, and  $w_{NT}(\alpha_i,\gamma_t)=w_{it}$, when we want to emphasize the fact that they are transformations of the underlying latent random variables. 
The following lemma follows directly from  \cite{eagleson1979orthogonal}.

\begin{lemma}[Orthonormal Representation of Degenerate Two-Sample U-Statistics]\label{orthonormal_representation}
	If $E[w_{NT}(\alpha_i,\gamma_t)^2]$ $<\infty$, then there exist complete orthonormal systems of square-integrable basis functions $\phi_{NT0}(\alpha)=1$, $\phi_{NT1}(\alpha)$,$\cdots$, and $l_{NT0}(\gamma)=1$, $l_{NT1}(\gamma)$,$\cdots$, such that
	\begin{align*}
	w_{it}=\sum_{j=0}^\infty \lambda_{NTj} \phi_{NTj}(\alpha_i)l_{NTj}(\gamma_t),\quad \sum_{j=0}^\infty \lambda_{NTj}^2<\infty.
	\end{align*}
	Furthermore,
	\begin{align*}
	E[\phi_{NTj}(\alpha_i) w_{NT}(\alpha_i,\cdot)]=\lambda_{NTj}l_{NTj}(\cdot),\quad E[l_{NTj}(\gamma_t)w_{NT}(\cdot,\gamma_t) ]=\lambda_{NTj}\phi_{NTj}(\cdot), j=1,2,\cdots
	\end{align*}
\end{lemma}

\noindent
By the orthonormality of the system,
\begin{align*}
	&\lambda_{NT0}=E[w_{NT}(\alpha_i,\gamma_t)]=0,
	&&E[\phi_{NTj}(\alpha_i)]=E[l_{NTj}(\gamma_t)]=0, \\
	&E[\phi_{NTj}^2(\alpha_i)]=E[l_{NTj}^2(\gamma_t)]=1,
	&&E[\phi_{NTj}(\alpha_i) l_{NTj}(\gamma_t)]=0
\end{align*}
hold for all $j=1,2,\cdots$.

Lemma \ref{orthonormal_representation} and Equation (\ref{eq:Hoeffding}) together imply that the DGP introduced in Equation (\ref{eq:when}) in Section \ref{sec:when} for an informal overview is indeed highly representative.

We impose the following conditions. Let us follow the mathematical convention $0/0=0$.

\begin{assumption}\label{a:deg_U-stat}
	(i) 
	$(\alpha_i)_{i\in \mathbb N}$ and 	$(\gamma_t)_{t\in \mathbb N}$ are i.i.d. Borel-measurable random vectors that are at most countable dimensional, and $(\alpha_i)_{i\in \mathbb N}\perp (\gamma_t)_{t\in \mathbb N}$. 
	(ii)
	$E[w_{NT}(\alpha_1,\gamma_1)]=0$, $E[w_{NT}^2(\alpha_1,\gamma_1)]<\infty$, $E[w_{NT}(\alpha_1,\gamma_1)|\alpha_1]=0$, $E[w_{NT}(\alpha_1,\gamma_1)|\gamma_1]=0$. 	   (iii)
	 The sample size satisfies $N\sim T$ as both of them diverge to infinity,\footnote{This condition not crucial for the theory and can be relaxed at a cost of more complicated rate conditions.}
	\begin{align}
	&\frac{N^{-1} E[w_{NT}^4(\alpha_1,\gamma_1)]}{\{E[w_{NT}^2(\alpha_1,\gamma_1)]\}^2}=o(1),\label{eq:Lyapunov}\\
	&\frac{E\big[(E[w_{NT}(\alpha_1,\gamma_1)w_{NT}(\alpha_2,\gamma_1)|\alpha_1,\alpha_2])^2\big] + E\big[(E[w_{NT}(\alpha_1,\gamma_1)w_{NT}(\alpha_1,\gamma_2)|\gamma_1,\gamma_2])^2\big]}{\{E[w_{NT}^2(\alpha_1,\gamma_1)]\}^2}=o(1).\label{eq:Hall}
	\end{align}
	
\end{assumption}

Assumption \ref{a:deg_U-stat} (i) is standard in the literature that uses exchangeable arrays to model two-way clustering -- see \citet{menzel2021bootstrap} and \citet{DDG2019}, for example, for more discussion on exchangeable models. 
Assumption \ref{a:deg_U-stat} (ii) states that the $U$-statistic of interest is degenerate and has at least two moments. 
In the setting of two-way clustering, this condition does not impose any restriction, but comes naturally from the property of the Hoeffding decomposition (\ref{eq:Hoeffding}). 
Assumption \ref{a:deg_U-stat} (iii) imposes the same growth rate between $N$ and $T$.
This condition can be relaxed at the cost of more complicated moment restrictions in the conditional moments of $w_{NT}$. 
Part (iii) further imposes a standard Lyapunov-type condition (\ref{eq:Lyapunov}), as well as (\ref{eq:Hall}), a condition that is analogous to the second half of Condition (2.1) in \cite{hall1984central}. 
This Hall-type condition restricts the size of the second moment of the conditional cross-products, $E[w_{NT}(\alpha_1,\gamma_1)w_{NT}(\alpha_2,\gamma_1)|\alpha_1,\alpha_2]$ and $E[w_{NT}(\alpha_1,\gamma_1)w_{NT}(\alpha_1,\gamma_2)|\gamma_1,\gamma_2]$, relative to the size of the second moment of $w_{NT}$. See Remark \ref{rem:Hall} below for a detailed discussion on its implications.	

\begin{assumption}\label{a:two-way}
	(i) $E[D_{11}^2]\ge\kappa_{\min}>0$. (ii) In addition,
	\begin{align*}
	&\frac{E[a_{NT}(\alpha_1)^4]+ E[b_{NT}(\gamma_1)^4]}{N\{(E[a_{NT}(\alpha_1)^2])^2 +  (E[b_{NT}(\gamma_1)^2])^2 \}+ N^3 (E[w_{NT}(\alpha_1,\gamma_1)^2])^2}=o(1) \text{ and }\\
	&\frac{N\{E\left[(a_{NT}(\alpha_1)  w_{NT}(\alpha_1,\gamma_2))^2\right] + E\left[(b_{NT}(\gamma_1)  w_{NT}(\alpha_2,\gamma_1))^2\right]\}}{(E[a_{NT}(\alpha_1)^2])^2 +  (E[b_{NT}(\gamma_1)^2])^2 + N^2(E[w_{NT}(\alpha_1,\gamma_1)^2])^2}=o(1).
	\end{align*}
\end{assumption}

The first condition in Assumption \ref{a:two-way} imposes a standard Lyapunov condition on the leading terms in the Hoeffding decomposition (\ref{eq:Hoeffding}). The second condition in Assumption \ref{a:two-way} is a mild restriction on how the linear term and quadratic terms in the Hoeffding decomposition (\ref{eq:Hoeffding}) can correlate in second moments, relatively to the variances of each component. Note that, by construction, the linear and quadratic components are uncorrelated. This condition is analogous to condition (1.6) in \cite{eubank1999central}.

\begin{remark}\label{rem:Hall}
	By the orthonormal representation, \cite{khashimov1989limit} points out that the condition
	\begin{align*}
	\frac{E\big[(E[w_{NT}(\alpha_1,\gamma_1)w_{NT}(\alpha_2,\gamma_1)|\alpha_1,\alpha_2])^2\big] + E\big[(E[w_{NT}(\alpha_1,\gamma_1)w_{NT}(\alpha_1,\gamma_2)|\gamma_1,\gamma_2])^2\big]}{\{E[w_{NT}^2(\alpha_1,\gamma_1)]\}^2}=o(1)
	\end{align*}
	in Assumption \ref{a:deg_U-stat} (iii) is equivalent to
	\begin{align*}
	\frac{\sum_{j=1}^\infty |\lambda_{NTj}|^4}{\{E[w_{NT}^2(\alpha_1,\gamma_1)]\}^2}=o(1).
	\end{align*}
	Simplifying it further, we have this condition equivalent in turn to
	\begin{align*}
	\frac{\sum_{j=1}^\infty \lambda_{NTj}^4}{\left(\sum_{j=1}^\infty \lambda_{NTj}^2E[\phi_{NTj}^2(\alpha_1)l_{NTj}^2(\gamma_1)]\right)^2}=o(1).
	\end{align*}
	If the unknown orthonormal basis  satisfies that $E[\phi_{NTj}^2(\alpha_1)l_{NTj}^2(\gamma_1)]$ is bounded and bounded away from zero for all  $j$'s with $\lambda_{NTj}\ne 0$, then this condition further reduces to
	\begin{align*}
	\frac{\sum_{j=1}^\infty \lambda_{NTj}^4}{\left(\sum_{j=1}^\infty \lambda_{NTj}^2\right)^2}=o(1).
	\end{align*}
	If the first $J$ has $\lambda_{NTj}\ne 0$, then for large $J$, it is more plausible for
	\begin{align*}
	\frac{\sum_{j=1}^J \lambda_{NTj}^4}{\left(\sum_{j=1}^J \lambda_{NTj}^2\right)^2}
	=\frac{\sum_{j=1}^J \lambda_{NTj}^4}{\sum_{j=1}^J\sum_{j'=1}^J \lambda_{NTj}^2 \lambda_{NTj'}^2}
	\end{align*}		
	to be small since the numerator is a sum of $J$ positive terms while the denominator is a sum over $J^2$ positive terms.
	This observation has an implication for the type of sequences of interactive fixed-effect models that are permitted. 
	For example, if all the factors with non-zero eigenvalues have equal weights, then this condition suggests that models with a large number of factors can be well-approximated by gaussian limiting distributions. 
	\qed
\end{remark}

\begin{theorem}[Gaussian Approximation for Arrays of Two-Way Clustering]\label{thm:CLT_two-way}
Suppose that Assumptions \ref{a:deg_U-stat} and \ref{a:two-way} are satisfied, and the limit $\lim_{N\wedge T \to \infty }Var(L_{NT})/Var(W_{NT})$ exists in $[0,\infty]$. Then, we have $\sigma_{LW,NT}^{-1}(L_{NT}+W_{NT}) \stackrel{d}{\to} N(0,1)$,
where $\sigma_{LW,NT}^2=Var(L_{NT}+W_{NT})=Var(L_{NT})+Var(W_{NT})$.	
\end{theorem}

This theorem implies that even the potentially non-gaussian term $W_{NT}$ alone in the decomposition \eqref{eq:Hoeffding} can be gaussian.
The following example illustrates a case in point.

\begin{example}
Consider a sequence of interactive fixed effects models
\begin{align*}
D_{it}=\sum_{j=1}^\infty \lambda_{NTj} \alpha_{ij}\gamma_{tj},
\end{align*}
where $\alpha_i=(\alpha_{ij})_{1\le i\le N,j\in \mathbb N}$ and  $\gamma_t=(\gamma_{tj})_{1\le t \le T, j\in \mathbb N}$ are mutually independent stochastic processes that have zero mean and Gaussian marginal distribution for each $i$, $t$, and $j$, i.e. $\alpha_{ij}\sim N(0,1)$, $\gamma_{tj}\sim N(0,1)$, and $\lambda_{NTj}$ is a sequence of constants for each $N,T$. 
Note that $\hat\theta_{NT}=(NT)^{-1}\sumi\sumt D_{it}$ in this example consists only of the $W_{NT}$ term in the decomposition \eqref{eq:Hoeffding}.
By the theorem, 
$	
\hat\theta_{NT}
$
is asymptotically gaussian if 
$\alpha_{ij}\perp \alpha_{i'j}$, $\gamma_{tj}\perp \gamma_{t'j}$, and
$
{\sum_{j=1}^\infty |\lambda_{NTj}|^4}/{\{Var(D_{11})\}^2}=o(1).
$
	\qed
\end{example}

We now turn to the asymptotic gaussianity of the whole $\hat \theta_{NT}$. 
From (\ref{eq:Hoeffding}), observe that only the $R_{NT}$ component remains random conditionally on all $\alpha_i$'s and $\gamma_t$'s, and is conditionally independent over $(i,t)$. 
An application of the Lyapunov CLT therefore implies that $R_{NT}$ is conditionally asymptotically gaussian given $\alpha_i$'s and $\gamma_t$'s. 
Thus, by the asymptotic gaussianity of $L_{NT}+W_{NT} $ from Theorem \ref{thm:CLT_two-way}, applying Theorem 1 in \cite{chen2007asymptotic} yields the asymptotic gaussianity of $\hat\theta_{NT}$. 
This conclusion is summarized as a corollary below.
\begin{corollary}\label{cor:CLT_two-way}
Suppose that the same set of conditions as in Theorem \ref{thm:CLT_two-way} is satisfied. If $E[|D_{11}|^3]<\infty$ holds in addition,
then
$
\sigma^{-1}\hat \theta_{NT}\stackrel{d}{\to} N(0,1),
$
where $\sigma_{NT}^2=Var(L_{NT})+Var(W_{NT})+Var(R_{NT})$.
\end{corollary}
\begin{remark}
The asymptotic gaussianity presented here is related to but differs in nature from those results for sparse networks graphs \citep[e.g.,][]{bickel2011method,graham2022kernel}. They rely on assumptions on the sequence of DGPs in which the $R_{NT}$ component is guaranteed to dominate the $W_{NT}$ component asymptotically, and thus their statistics are asymptotically gaussian regardless of whether the $W_{NT}$ component is gaussian or not.
In contrast, we guarantee the gaussianity of the potentially non-gaussian component $W_{NT}$, and hence the gaussianity of the whole $\hat\theta_{NT}$ follows without relying on the specific DGPs that have the $R_{NT}$ term dominate the $W_{NT}$ term.
\qed
\end{remark}

\section{Simulations}\label{sec:simulations}

We focus on the data generating processes that correspond to the potentially non-gaussian components in \citet[][Section 2.2]{menzel2021bootstrap}, i.e., the completely degenerate two-sample U-statistics, and their variants.
Such data generating processes are derived from the orthonormal representation in Section \ref{sec:main}.
Specifically, we generate a sample $\{D_{it}: 1 \le i \le N, 1 \le t \le T \}$ by
\begin{align*}
D_{it} = J^{-1} \sum_{j=1}^J (\alpha_{ij} - \delta) (\gamma_{tj} - \delta) + \phi \varepsilon_{it},
\end{align*}
where $\alpha_{ij} \sim N(0,1)$, $\gamma_{tj} \sim N(0,1)$, and $\varepsilon_{it} \sim N(0,1)$ are mutually independent and independent across $i \in \{1,\cdots,N\}$, $t \in \{1,\cdots,T\}$ and $j \in \{1,\cdots,J\}$. Note that this DGP is a special case of the generic DGP of Equation (\ref{eq:when}) with $\lambda_{j}=J^{-1}$.
We vary the data generating parameters $\delta$, $J$, and $\phi$ across sets of Monte Carlo simulations.
Each set of simulations consists of 10,000 random draws of $\{D_{it}: 1 \le i \le N, 1 \le t \le T \}$.
The sample size is set to $N=T=50$ throughout.

For the population mean
$
\theta = E[D_{it}]
$
as the parameter of interest, consider the two-way sample mean estimator
$
\hat\theta_{NT} = (NT)^{-1} \sum_{i=1}^N \sum_{t=1}^T D_{it}.
$
Note that our theory from Section \ref{sec:main} predicts that the gaussian approximation is asymptotically reasonable unless both $|\delta|$ and $J$ are too small.
For each draw, we construct the 95\% confidence intervals by two existing methods:
one is based on the widely used two-way cluster-robust standard error \citep[e.g.,][]{CGM2011,thompson2011simple} which presumes the asymptotic gaussianity;
and the other is the recent bootstrap method by \citet{menzel2021bootstrap} which does not require the asymptotic gaussianity.
We will hereafter refer to the first method by `CGM' and the second method by `M' for brevity.

Figure \ref{fig:qq_fix_phi} illustrates QQ plots of $\hat\theta_{NT}$ for $\delta \in \{0.0,0.5,1.0\}$, $J \in \{1,50,100\}$ and $\phi = 0.5$.
The top left plot is associated with the most pathetic case with small $|\delta|$ and small $J$, for which our theory cannot guarantee that the gaussian approximation is asymptotically reasonable.
This QQ plot shows that the sample quantiles of $\hat\theta_{NT}$ indeed fail to conform with the theoretical quantiles.
On the other hand, when $|\delta|$ or $J$ takes a larger value, the gaussian approximation becomes more reasonable.
Specifically, the remaining eight QQ plots in Figure \ref{fig:qq_fix_phi} show that the sample quantiles are sufficiently close to the theoretical quantiles.
These results support our theoretical prediction from Section \ref{sec:main}.

\begin{figure}
\vspace{0.5cm}
\begin{tabular}{ccc}
\includegraphics[width=0.3\textwidth]{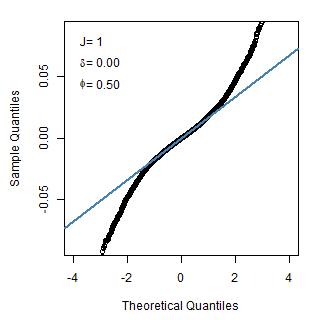} &
\includegraphics[width=0.3\textwidth]{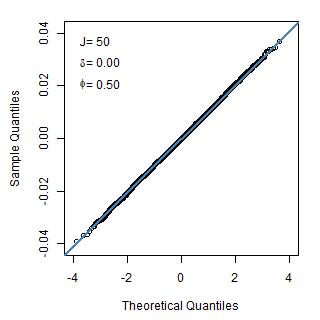} &
\includegraphics[width=0.3\textwidth]{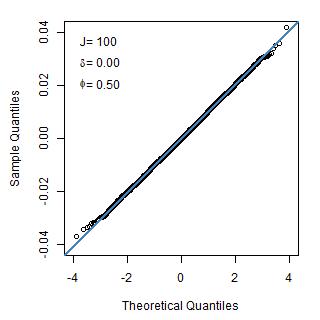}	
\\
\includegraphics[width=0.3\textwidth]{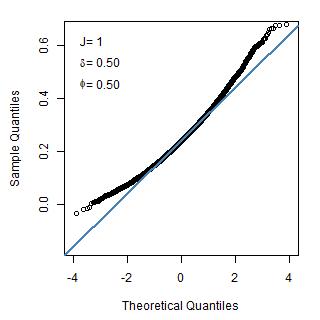} &
\includegraphics[width=0.3\textwidth]{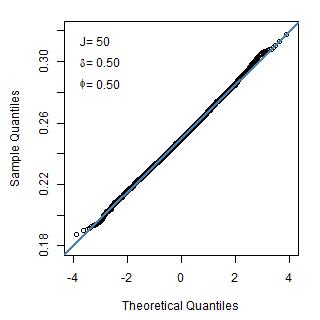} &
\includegraphics[width=0.3\textwidth]{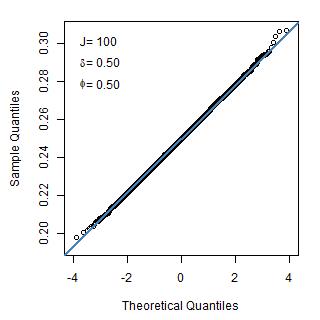}	
\\
\includegraphics[width=0.3\textwidth]{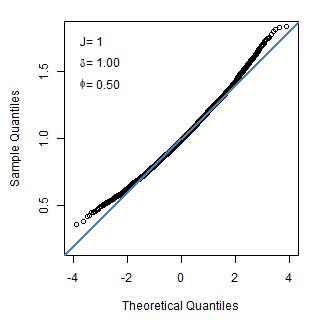} &
\includegraphics[width=0.3\textwidth]{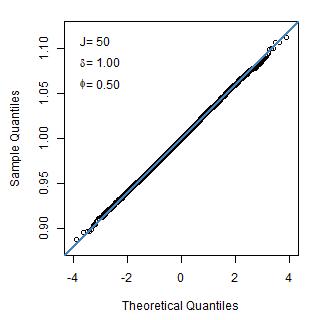} &
\includegraphics[width=0.3\textwidth]{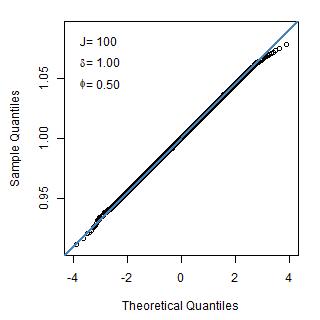}	
\end{tabular}
\caption{QQ plots for $\delta \in \{0.0,0.5,1.0\}$, $J \in \{1,50,100\}$ and $\phi = 0.5$. The results are based on 10,000 Monte Carlo iterations.}\label{fig:qq_fix_phi}${}$
\end{figure}

Figure \ref{fig:cover_fix_phi} illustrates coverage frequencies as a function of $\phi \in [0.0,1.0]$ for $J \in \{1,50,100\}$ and $\delta \in \{0.0,0.5,1.0\}$, with the nominal coverage probability of 95\%.
The positions of the nine coverage plots in Figure \ref{fig:cover_fix_phi} correspond to those of the nine QQ plots in Figure \ref{fig:qq_fix_phi}.
Thus, the top left plot is associated with the most pathetic case in which $\hat\theta_{NT}$ is far away from gaussian.
In this plot, `CGM' suffers from severe under-coverage and `M' in contrast suffers from over-coverage.
The under-coverage by `CGM' can be explained by the non-gaussianity of $\hat\theta_{NT}$.

\begin{figure}
\vspace{0.5cm}
\begin{tabular}{ccc}
\includegraphics[width=0.3\textwidth]{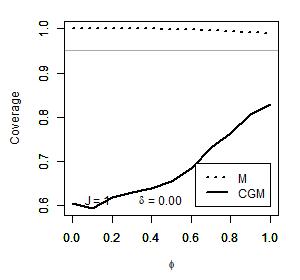} &
\includegraphics[width=0.3\textwidth]{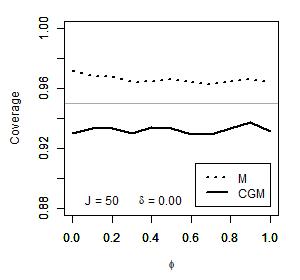} &
\includegraphics[width=0.3\textwidth]{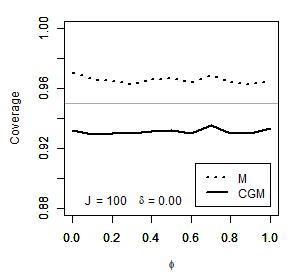}	
\\
\includegraphics[width=0.3\textwidth]{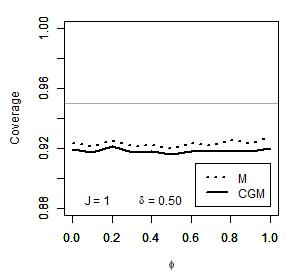} &
\includegraphics[width=0.3\textwidth]{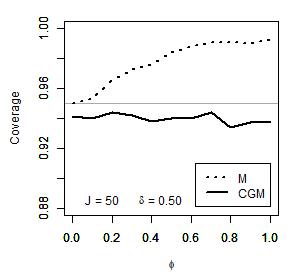} &
\includegraphics[width=0.3\textwidth]{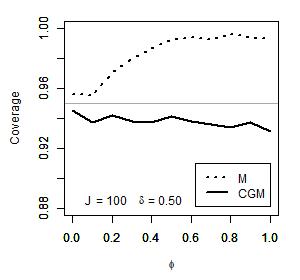}	
\\
\includegraphics[width=0.3\textwidth]{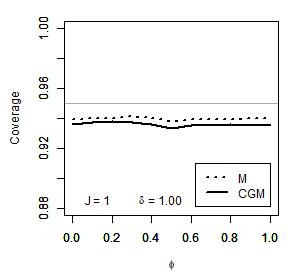} &
\includegraphics[width=0.3\textwidth]{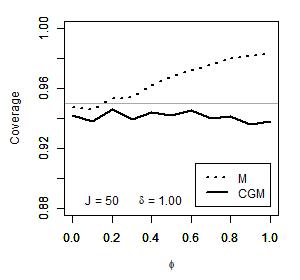} &
\includegraphics[width=0.3\textwidth]{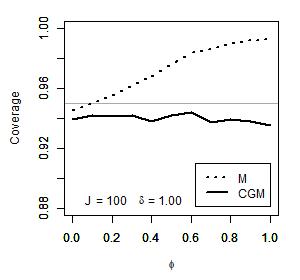}	
\end{tabular}
\caption{Coverage frequencies as a function of $\phi \in [0.0,1.0]$ for $J \in \{1,50,100\}$, $\delta \in \{0.0,0.5,1.0\}$. The nominal probability is 95\%. The results are based on 10,000 Monte Carlo iterations.}\label{fig:cover_fix_phi}${}$
\end{figure}

In the remaining eight plots in Figure \ref{fig:cover_fix_phi}, on the other hand, `CGM' achieves fairly precise coverage.
Again, as predicted by our theory from Section \ref{sec:main} and also evidenced by the QQ plots in Figure \ref{fig:qq_fix_phi}, these remaining eight cases admit reasonably close gaussian approximations for the distribution of $\hat\theta_{NT}$.
Therefore, these precise coverage results by `CGM'  in these eight cases meet our expectations.

Figure \ref{fig:cover_fix_phi} discretely varies $J$ and $\phi$ while it continuously varies $\phi$.
We next change our perspectives by continuously varying $\delta$.
Figure \ref{fig:cover_fix_delta} illustrates coverage frequencies as a function of $\delta \in [0.0,1.0]$ for $J \in \{1,50,100\}$ and $\phi \in \{0.0,0.5,1.0\}$.
As before, the nominal probability is 95\%.

\begin{figure}
\vspace{0.5cm}
\begin{tabular}{ccc}
\includegraphics[width=0.3\textwidth]{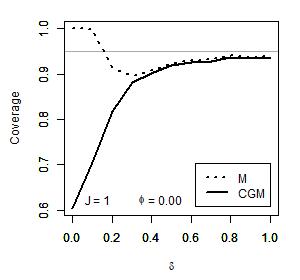} &
\includegraphics[width=0.3\textwidth]{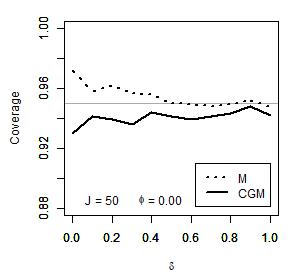} &
\includegraphics[width=0.3\textwidth]{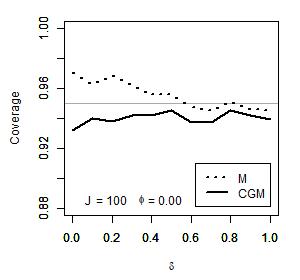}	
\\
\includegraphics[width=0.3\textwidth]{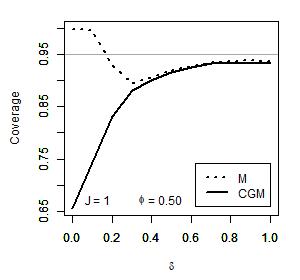} &
\includegraphics[width=0.3\textwidth]{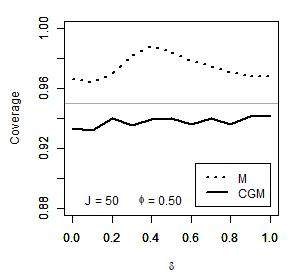} &
\includegraphics[width=0.3\textwidth]{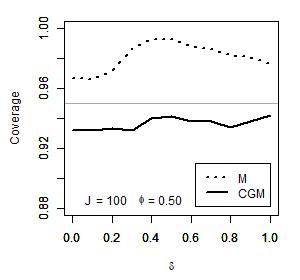}	
\\
\includegraphics[width=0.3\textwidth]{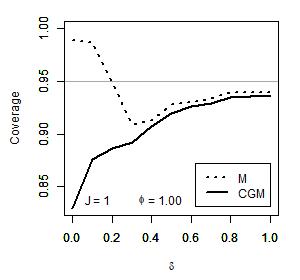} &
\includegraphics[width=0.3\textwidth]{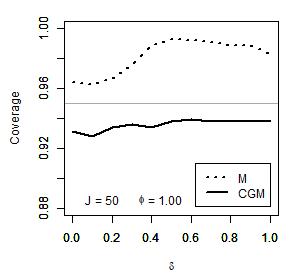} &
\includegraphics[width=0.3\textwidth]{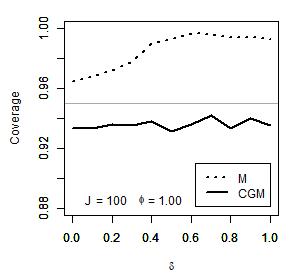}	
\end{tabular}
\caption{Coverage frequencies as a function of $\delta \in [0.0,1.0]$ for $J \in \{1,50,100\}$ and $\phi \in \{0.0,0.5,1.0\}$. The nominal probability is 95\%. The results are based on 10,000 Monte Carlo iterations.}\label{fig:cover_fix_delta}${}$
\end{figure}

On the three plots in the left column of Figure \ref{fig:cover_fix_delta}, where we set $J=1$, `CGM' suffers from under-coverage and `M' suffers from over-coverage in the region where $\delta$ takes small values.
As $\delta$ increases, however, their coverage accuracy improves.
On the remaining six plots in the middle and right columns of Figure \ref{fig:cover_fix_delta}, where we set $J=50$ and $100$, respectively, `CGM' achieves fairly accurate coverage regardless of the values taken by $\delta$.
These results again conform with our theoretical prediction of asymptotic gaussianity, which holds unless both $|\delta|$ and $J$ are too small.

Thus, far, we have analyzed the simulation results by continuously varying $\phi$ (Figure \ref{fig:cover_fix_phi}) and $\delta$ (Figure \ref{fig:cover_fix_delta}).
We next present simulation results with finer variations of $J$.
Figure \ref{fig:cover_fix_J} illustrates coverage frequencies as a function of $J \in \{1,\cdots,100\}$ for $\delta \in \{0.0,0.5,1.0\}$ and $\phi \in \{0.0,0.5,1.0\}$.

\begin{figure}
\vspace{0.5cm}
\begin{tabular}{ccc}
\includegraphics[width=0.3\textwidth]{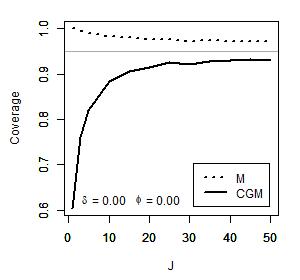} &
\includegraphics[width=0.3\textwidth]{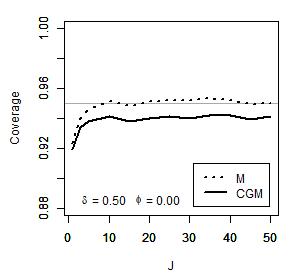} &
\includegraphics[width=0.3\textwidth]{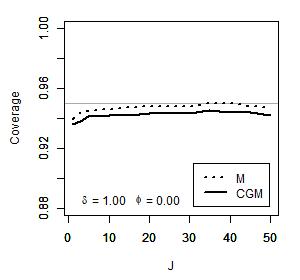}	
\\
\includegraphics[width=0.3\textwidth]{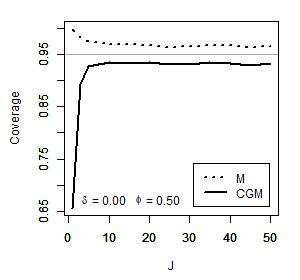} &
\includegraphics[width=0.3\textwidth]{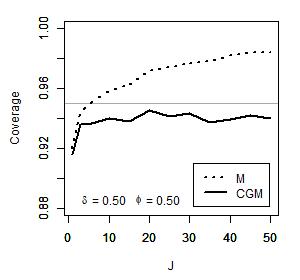} &
\includegraphics[width=0.3\textwidth]{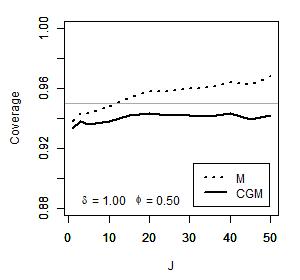}	
\\
\includegraphics[width=0.3\textwidth]{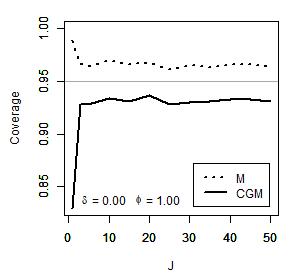} &
\includegraphics[width=0.3\textwidth]{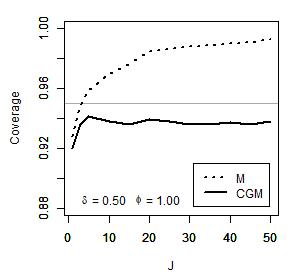} &
\includegraphics[width=0.3\textwidth]{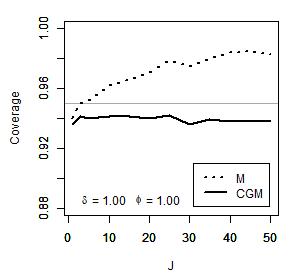}	
\end{tabular}
\caption{Coverage frequencies as a function of $J \in \{1,\cdots,100\}$ for $\delta \in \{0.0,0.5,1.0\}$ and $\phi \in \{0.0,0.5,1.0\}$. The nominal probability is 95\%. The results are based on 10,000 Monte Carlo iterations.}\label{fig:cover_fix_J}${}$
\end{figure}

On the three plots in the left column of Figure \ref{fig:cover_fix_J}, where we set $\delta=0.00$, `CGM' suffers from under-coverage and `M' suffers from over-coverage in the region where $J$ takes small values.
As $J$ increases, however, their coverage accuracy improves.
On the remaining six plots in the middle and right columns of Figure \ref{fig:cover_fix_J}, where we set $\delta=0.50$ and $1.00$, respectively, `CGM' achieves fairly accurate coverage regardless of the values of $J$.
These results again conform with our theoretical prediction of asymptotic gaussianity, which holds unless both $|\delta|$ and $J$ are too small. 
We emphasize that gaussian limiting distributions still provide good approximations in most cases even when $\phi=0$ and $\delta=0$, i.e. even when only the degenerate two-sample $U$-statistics are present in the DGP.

In summary, the asymptotic gaussian approximation appears reasonable except for the pathetic case in which both $|\delta|$ and $J$ are too small, as predicted by our theory from Section \ref{sec:main}.
Consequently, the conventional two-way cluster-robust standard errors \citep[e.g.,][]{CGM2011,thompson2011simple}, which presume the asymptotic gaussianity, perform sufficiently well unless both $|\delta|$ and $J$ are too small.
Recall that we have already focused on the data generating processes that correspond to the potentially non-gaussian components in \citet[][Section 2.2]{menzel2021bootstrap}.
Even within this class with the potential non-gaussianity, we confirm that gaussianity is in fact rather common than exceptional.

\section{Discussions}

Thanks to the earlier work in the literature, we now have a more complete understanding of the settings and conditions under which the asymptotic gaussianity is guaranteed.
To our knowledge, there are three other scenarios in addition to our conditions in which the gaussianity holds under two-way clustering. 
Let us maintain $N\sim T$.
Table \ref{tab:conditions} summarizes them.

\begin{table}[t]
	\centering	
	\begin{tabular}{c|ccccc}
		Conditions &  Consequence&Reference\\
		\hline
		\hline
		Non-degenerate&
	$L_{NT}$ dominates& \cite{DDG2019}
		 \\
		Sparse network&$L_{NT}+R_{NR}$ dominates&\cite{graham2022sparse} \\
		Kernel density estimation& $L_{NT}+R_{NR}$ dominates &\cite{graham2022kernel}\\
		Our conditions&None of the term needs to dominate& This paper\\
		\hline
	\end{tabular}
	\caption{Conditions that guarantee the gaussianity under two-way clustering.}${}$
	\label{tab:conditions}
\end{table}
 
The most well-known case is the ``non-degeneracy'' condition imposed in e.g. \cite{DDG2019}, which assumes that at least one of $E[D_{it}|\alpha_i]$ and $E[D_{it}|\gamma_t]$ is random with its variance bounded away from zero. 
This condition imposes that at least one of the latent cluster-specific shocks has an impact on the level of the conditional mean. 
If this condition is satisfied, then $\hat\theta_{NT}\approx L_{NT}$ and the asymptotic gaussianity holds regardless of how the $W_{NT}$ term behaves. 
  
The next case is illustrated by the ``sparse network'' asymptotics, which was first considered in \cite{bickel2011method} for dyadic network graphs, and studied under two-way clustering by \cite{graham2022sparse}. 
In this case, one assumes $D_{it}\in\{0,1\}$ and $p:=P(D_{it}=1)$ is converging to zero or one. 
Under this setting, one can exploit the binary nature of the outcome variable and show $W_{NT}=O_p(p/N)$ to be asymptotically dominated by $R_{NT}=O_p(p^{1/2}/N)$, and thus $\hat \theta_{NT}\approx L_{NT}+R_{NT}$. 
Therefore, the asymptotic gaussianity is guaranteed by an application of a martingale CLT. 
The third case is related to the second one and considers kernel density estimation for two-way clustered random variables, studied in \cite{graham2022kernel}.
 
In light of these conditions for the asymptotic gaussianity along with our CLT result, we propose the decision tree shown in Figure \ref{fig:flow_chart} to researchers considering adoption of the TWCR standard errors for inference.

\begin{figure}[t]
\centering
\begin{tikzpicture}
[sibling distance=15em,
every node/.style = {shape=rectangle, rounded corners,
	draw, align=center,
	top color=white, bottom color=blue!20}]]
  \node {Non-degenerate?\\
  	(As assumed by \cite{DDG2019})}
child { node {Yes} 
	child{ node {TWCR is valid}	}
}
       child { node {No} 
       	child{ node{Sparse network?\\
       			(As assumed by \cite{graham2022sparse})}
       		child{ node{Yes} child{ node{TWCR is valid}}}
       	child{ node{No}child{ node{Very few latent factors?\\
       				(As ruled out by this paper)} child{ node{Yes} child{node{TWCR is not valid}}}
       	child{node{No} child{node{TWCR is still valid}}}   
}
}
}
};
\end{tikzpicture}
\caption{A decition tree for researchers considering an adoption of the two-way cluster-robust (TWCR) standard errors for inference.}${}$
\label{fig:flow_chart}
\end{figure}
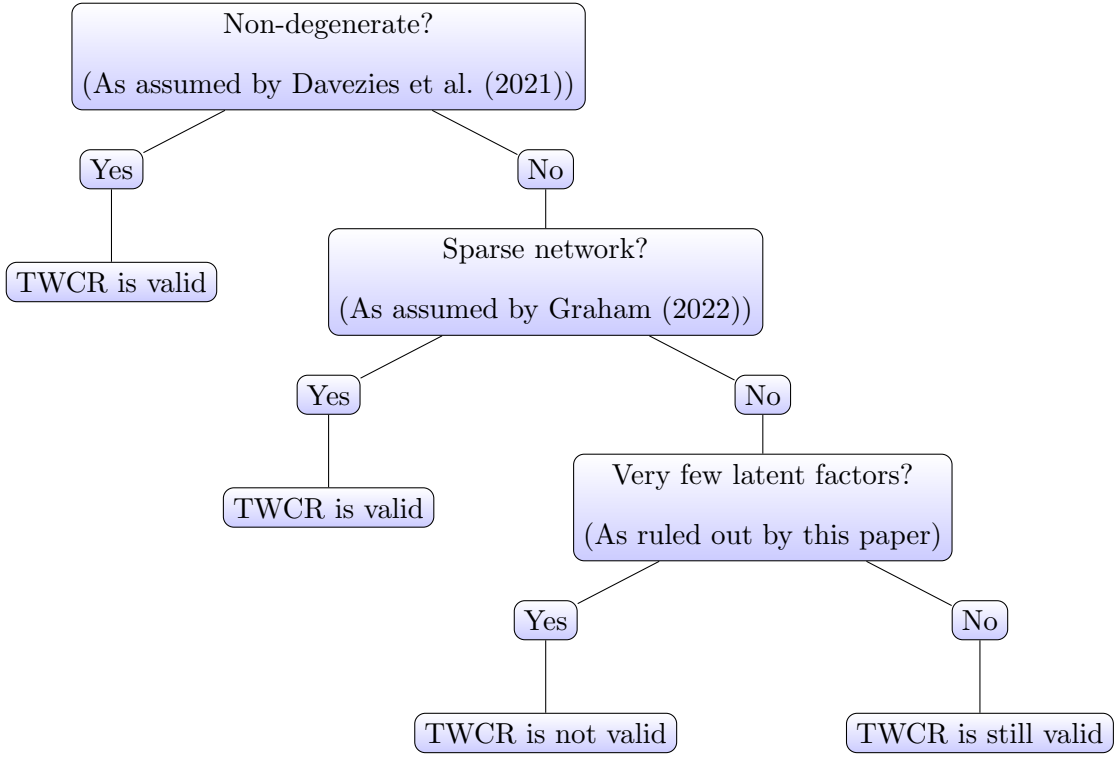

First, decide whether the researcher is willing to assume the non-degeneracy or a sparse network.
If the answer is yes, then one can use the TWCR standard errors.
Otherwise, ask if the model under consideration involves only very few latent factors.
If the answer is no, then one can still use the TWCR standard errors.
In many economic applications, researchers would rather want to suppose that there are many latent factors in economic agents, e.g., many dimensions of ability measures, many dimensions of health measures, and many dimensions of cultural attributes.
Our theoretical result endorses the validity of the TWCR standard errors in these applications that are common in the economic research.

In summary, we conclude that the two-way cluster-robust (TWCR) standard errors, which were proposed by \citet{CGM2011} and \citet{thompson2011simple} and have been used by thousands of empirical research papers, are valid under most, if not all, circumstances of the economic research.

\vspace{0.5cm}
\appendix
\section*{Appendix}

\section{Auxiliary Results}
	\begin{lemma}\label{lem:nice}
	If $E[w_{NT}(\alpha_1,\gamma_1)|\alpha_1]=0$, $E[w_{NT}(\alpha_1,\gamma_1)|\gamma_1]=0$, then 
	\begin{align*}
	E[w_{i_1t_1}w_{i_2t_2}...w_{i_kt_k}]=0
	\end{align*}
	if at leat one of $i\in\{i_1,...,i_k\}$ or $t\in\{t_1,...,t_k\}$ appears only once.
	\begin{proof}
		Consider the case in which $i_1$ appears only once.
		In this case, we have
		\begin{align*}
		E[w_{i_1t_1}w_{i_2t_2}...w_{i_kt_k}]
		=&
		E[E[w_{i_1t_1}w_{i_2t_2}...w_{i_kt_k}|\alpha_{i_2},...,\alpha_{i_k}, \gamma_{t_1},...,\gamma_{t_k}]]\\
		=&
		E[w_{i_2t_2}...w_{i_kt_k}E[w_{i_1t_1}|\alpha_{i_2},...,\alpha_{i_k}, \gamma_{t_1},...,\gamma_{t_k}]]\\
		=&
		E[w_{i_2t_2}...w_{i_kt_k}\underbrace{E[w_{i_1t_1}| \gamma_{t_1}]}_{=0 }].
		\end{align*} 
		The statement in the other case can be shown analogously.
	\end{proof}
\end{lemma}

The following Theorem shows a CLT for the two-sample degenerate $U$-statistic $W_{NT}$ component in the decomposition (\ref{eq:Hoeffding}). 
\begin{lemma}[Central Limit Theorem for Degenerate Two-Sample U-Statistics]\label{thm:CLT_U-stat}
If Assumption \ref{a:deg_U-stat} is satisfied, then $\sigma_{W,NT}^{-1}W_{NT} \stackrel{d}{\to} N(0,1)$, where $\sigma_{W,NT}^2=Var(W_{NT})$.	
\end{lemma}

\begin{remark}
In fact, the statement of the lemma holds not only for the $W_{NT}$ component in the decomposition (\ref{eq:Hoeffding}), but also for any generic completely degenerate two-sample $U$-statistics that satisfy Assumption \ref{a:deg_U-stat}. 
Although such a CLT result is not new as it was already shown in \cite{khashimov1989limit}, the proof presented here is new and relies on the martingale CLT of \cite{heyde1970departure}. 
In contrast, \cite{khashimov1989limit} -- see also \cite{khashimov1989limit2} -- relies on the orthonormal representation as well as approximating the characteristic function of $W_{NT}$ by the characteristic function of $ \sumi\sumt\sum_{j=1}^\infty\lambda_{NTj}\xi_{ij}\zeta_{tj}$ for some i.i.d. standard normal random variables  $\xi_{ij}$ and $\zeta_{tj}$'s. 
Our alternative proof strategy using martingale theory enables us to generalize the result to show Theorem \ref{a:two-way}. 
The proof strategy here is related to \cite{hall1984central} and \cite{de1987central}, but the construction of martingale difference sequences is different. 
\qed
\end{remark}

\begin{proof}[Proof of Lemma \ref{thm:CLT_U-stat}]
For any set of random variables $A$, let $\sigma(A)$ denote the smallest $\sigma$-ring generated by $X$.
Assume $N\le T$ without loss of generality, and define\small
	\begin{align*}
	U_1=&\sigma_{W,NT}^{-1}w_{NT}(\alpha_1,\gamma_1), \quad \calF_0=\sigma(\{\emptyset\}),\\
	U_2=&\sigma_{W,NT}^{-1}\left(w_{NT}(\alpha_2,\gamma_1) + w_{NT}(\alpha_2,\gamma_2) + w_{NT}(\alpha_1,\gamma_2)\right), \quad \calF_1=\sigma( \alpha_1, \gamma_1),\\
	U_3=&\sigma_{W,NT}^{-1}\left(w_{NT}(\alpha_3,\gamma_1)  + w_{NT}(\alpha_3,\gamma_2)+w_{NT}(\alpha_3,\gamma_3) +w_{NT}(\alpha_2,\gamma_3) +w_{NT}(\alpha_1,\gamma_3)\right), \quad \calF_2=\sigma( \alpha_1, \alpha_2, \gamma_1, \gamma_2),\\
	\vdots\\
	U_N=&\sigma_{W,NT}^{-1}\left(\sum_{1\le t<N } w_{NT}(\alpha_N,\gamma_t) +  w_{NT}(\alpha_N,\gamma_N) + \sum_{1\le i <N } w_{NT}(\alpha_i,\gamma_N)\right),\quad
	\calF_{N-1}=\sigma( (\alpha_i)_{i=1}^{N-1}, (\gamma_t)_{t=1}^{N-1}),\\
	U_{N+1}=&\sigma_{W,NT}^{-1}\left(\sum_{1\le i\le N } w_{NT}(\alpha_i,\gamma_{N+1})\right),
	\quad \calF_N=\sigma((\alpha_i)_{i=1}^{N}, (\gamma_t)_{t=1}^{N}),\\
	\vdots\\
	U_{T}=&\sigma_{W,NT}^{-1}\left(\sum_{1\le i\le N } w_{NT}(\alpha_i,\gamma_{T})\right),\quad \calF_{T-1}=\sigma( (\alpha_i)_{i=1}^{N}, (\gamma_t)_{t=1}^{T-1}).
	\end{align*}\normalsize
Note that $\sum_{s=1}^T U_s=\sum_{i=1}^N \sum_{t=1}^T w_{NT}(\alpha_i,\gamma_t)=W_{NT}$.
Also, since $E[w_{NT}(\alpha_i,\gamma_t)]=E[w_{NT}(\alpha_i,\gamma_t)|\alpha_i]=E[w_{NT}(\alpha_i,\gamma_t)|\gamma_t]=0$, we have $E[U_s|\calF_{s-1}]=0$ for all $s=1,...,T$, i.e. the above forms a martingale difference sequence.\footnote{This construction seems to be new.}	
	
To proceed, we verify the following two conditions
\begin{align*}
	&\sum_{s=1}^T E[|U_s|^{4}]=o(1) \text{ as } T\to\infty,  &&\text{(Condition I)}\\
	&Var\left(\sum_{s=1}^T E[U_s^2|\calF_{s-1}]\right)=o(1) \text{ as } T\to \infty. &&\text{(Condition II)}
\end{align*}
Once we verify these two conditions, we can apply the martingale CLT of \cite{heyde1970departure} under their conditions (2) and (4). 

	\bigskip
	\noindent{\textbf{Verifying Condition I.}}
	For any $1<s\le N <T$, we can write
	\begin{align}
	E[U_s^4]=&E\Bigg[\sigma_{W,NT}^{-4}\Bigg(\sum_{1\le t<s } w_{st} + \sum_{1\le i \le s } w_{is}\Bigg)^4\Bigg]\nonumber\\
	=&
	\sigma_{W,NT}^{-4}\Bigg(\sum_{1\le t<s }\sum_{1\le t'<s }\sum_{1\le t''<s }\sum_{1\le t'''<s } E[w_{st} w_{st'}  w_{st''}  w_{st'''}] \nonumber\\
	&+ 
	4\sum_{1\le t<s }\sum_{1\le t'<s }\sum_{1\le t''<s }\sum_{1\le i \le s } E[w_{st}  w_{st'}  w_{st''} w_{is} ]\nonumber\\
	&+6 \sum_{1\le t<s }\sum_{1\le t'<s }\sum_{1\le i \le s }\sum_{1\le i' \le s } E[w_{st} w_{st'} w_{is} w_{i's}] \nonumber\\
	&+4\sum_{1\le t<s }\sum_{1\le i \le s }\sum_{1\le i' \le s }\sum_{1\le i'' \le s }E[w_{st} w_{is} w_{i's}w_{i''s}]\nonumber\\
	&+\sum_{1\le i \le s }\sum_{1\le i' \le s }\sum_{1\le i'' \le s }\sum_{1\le i''' \le s }E[w_{is}w_{i's}w_{i''s}w_{i'''s}]
	\Bigg).\label{eq:cond-I_fourth-moment-1}
	\end{align}
	For the first term on the RHS of \eqref{eq:cond-I_fourth-moment-1}, note by Lemma \ref{lem:nice} that $E[w_{st} w_{st'}  w_{st''}  w_{st'''}]\ne 0$ only if either the four $t$ indices are all the same or if they consist of two two-pairs. The second term on the RHS of \eqref{eq:cond-I_fourth-moment-1} is zero since $s$ in the second index appears only once in each of the  $E[w_{st}w_{st'}w_{st''}w_{is}]$ term.  The fourth term on the RHS of \eqref{eq:cond-I_fourth-moment-1} is zero since for each term in the sum, $t$ with $t<s$ appears only once.
	For the third term on the RHS of \eqref{eq:cond-I_fourth-moment-1}, we decompose\small
	\begin{align*}
	&\sum_{1\le t<s }\sum_{1\le t'<s }\sum_{1\le i \le s }\sum_{1\le i' \le s } E[w_{st} w_{st'} w_{is} w_{i's}]\\
	=&
	\sum_{1\le t<s }\sum_{1\le t'<s }\sum_{1\le i < s }\sum_{1\le i' < s } E[w_{st} w_{st'} w_{is} w_{i's}]
	+2\sum_{1\le t<s }\sum_{1\le t'<s }\sum_{1\le i < s } E[w_{st} w_{st'} w_{is} w_{ss}]
	+\sum_{1\le t<s }\sum_{1\le t'<s }E[w_{st} w_{st'} w_{ss}^2 ]\\
	=&
	\sum_{1\le t<s}\sum_{1\le i<s} E[w_{st}^2w_{is}^2]
	+0
	+\sum_{1\le t <s} E[w_{st}^2 w_{ss}^2]
	=\sum_{1\le t<s}\sum_{1\le i\le s} E[w_{st}^2w_{is}^2],
	\end{align*}\normalsize
	where the second term on the right of the first equality is zero by Lemma \ref{lem:nice} and the observation that each $i$ appears only once.
	
	Thus, (\ref{eq:cond-I_fourth-moment-1}) simplifies into
	\begin{align*}
	E[U_s^4]=&\sigma_{W,NT}^{-4}\Bigg(
	\sum_{1\le t< s} E[ w_{st}^4]
	+ \sum_{1\le t <s}\sum_{1 \le t' <s, t'\ne t} E[ w_{st}^2w_{st'}^2]
	+0\\
	&+6\sum_{1\le t<s}\sum_{1\le i\le s} E[w_{st}^2w_{is}^2]
	+0
	+\sum_{1\le i\le s} E[w_{is}^4] + \sum_{1\le i \le s }\sum_{ 1\le i'\le s, i' \ne i} E[w_{is}^2w_{i's}^2]
	\Bigg).
	\end{align*}
	Now, for $N<s\le T$, we have
	\begin{align*}
	U_s=\sigma_{W,NT}^{-1} \sum_{1\le i \le N} w_{i s},
	\end{align*}
	and thus
	\begin{align*}
	E[U_s^4]=&\sigma_{W,NT}^{-4} \sum_{1\le i \le N}  \sum_{1\le i' \le N} \sum_{1\le i'' \le N}  \sum_{1\le i''' \le N}E[w_{i s}w_{i' s}w_{i'' s}w_{i''' s}]
	\\
	=&\sigma_{W,NT}^{-4}\left( \sum_{1\le i \le N} E[w_{i s}^4]+\sum_{1\le i \le N}\sum_{1\le i' \le N}E[w_{i s}^2 w_{i' s}^2]\right).
	\end{align*}
	Therefore, we have
	\begin{align*}
	\sum_{s=1}^T E[U_s^4]
	=&\sigma_{W,NT}^{-4}\sum_{s=1}^N\Bigg(
	\sum_{1\le t< s} E[ w_{st}^4]
	+ \sum_{1\le t <s}\sum_{1\le t'<s, t'\ne t} E[ w_{st}^2w_{st'}^2]\\
	&\quad+6\sum_{1\le t<s}\sum_{1\le i\le s} E[w_{st}^2w_{is}^2]
	+\sum_{1\le i\le s} E[w_{is}^4] + \sum_{1\le i \le s }\sum_{ 1\le i'\le s, i' \ne i} E[w_{is}^2w_{i's}^2]
	\Bigg)\\
	&+\sigma_{W,NT}^{-4}\sum_{s=N+1}^T\left( \sum_{1\le i \le N} E[w_{i s}^4]+\sum_{1\le i \le N}\sum_{1\le i' \le N}E[w_{i s}^2 w_{i' s}^2]\right).
	\end{align*}
	By the identical distribution, the above expression further simplifies to
	\begin{align*}
	\sum_{s=1}^T E[U_s^4]
	=&\sigma_{W,NT}^{-4}\sum_{s=1}^N\Bigg(
	\sum_{1\le t\le s} E[ w_{11}^4]
	+ \sum_{1\le t <s}\sum_{t'\ne t}^s E[ w_{11}^2w_{12}^2]\\
	&\quad+6\sum_{1\le t<s}\sum_{1\le i\le s} E[w_{11}^2w_{22}^2]
	+\sum_{1\le i\le s} E[w_{11}^4] + \sum_{1\le i \le s }\sum_{ i' \ne i}^s E[w_{11}^2w_{21}^2]
	\Bigg)\\
	&+\sigma_{W,NT}^{-4}\sum_{s=N+1}^T\left( \sum_{1\le i \le N} E[w_{11}^4]+\sum_{1\le i \le N}\sum_{1\le i' \le N}E[w_{11}^2 w_{21}^2]\right)\\
	=&\sigma_{W,NT}^{-4}N(1+o(1))\Bigg(
	\frac{N}{2} E[ w_{11}^4]
	+  \frac{N^2}{4} E[ w_{11}^2w_{12}^2]\\
	&\quad+\frac{6N^2}{4} E[w_{11}^2w_{22}^2]
	+\frac{N}{2}E[w_{11}^4] + \frac{N^2}{4}E[w_{11}^2w_{21}^2]
	\Bigg)\\
	&+\sigma_{W,NT}^{-4}(T-N)\left( N E[w_{11}^4]+N^2E[w_{11}^2 w_{21}^2]\right)
	\lesssim \sigma_{W,NT}^{-4}N^3 E[w_{11}^4]
	\end{align*}
	when $N\sim T$.
	Now, recall $E[w_{it}]=0$ and note
	\begin{align*}
	\sigma_{W,NT}^2=&Var\left(\sumi\sumt w_{it}\right)
	=E\left[\left(\sumi\sumt w_{it}\right)^2\right]
	=\sumi\sumt E[w_{it}^2].
	\end{align*}
	Here, we once again use Lemma \ref{lem:nice} to conclude that any cross-moment is zero as it would contain either an $i$ or a $t$ that appears only once.
	Thus, we obtain
	\begin{align*}
	\sigma_{W,NT}^4=&\left(\sumi\sumt E[w_{it}^2]\right)^2=\sumi\sum_{i'=1}^N \sumt\sum_{t'=1}^T  E[w_{it}^2] E[w_{i't'}^2]= N^2T^2  \left\{E[w_{11}^2]\right\}^2
	\end{align*}
	Then, as $N\sim T$, the condition is satisfied if
	\begin{align*}
	\frac{E[w_{11}^4]}{N \{E[w_{11}^2]\}^2}=o(1),
	\end{align*}
	which holds under Assumption \ref{a:deg_U-stat}.
	We now conclude that
$
	\sum_{s=1}^T E[|U_s|^4]=o(1)
$
	and this verifies Condition I.

	\bigskip
	\noindent{\textbf{Verifying Condition II.}}
	We are going to show 
	\begin{align*}
	Var\left(\sum_{s=1}^T E[U_s^2|\calF_{s-1}]\right)=o(1).
	\end{align*}
	Note that
	\begin{align*}
	Var\left(\sum_{s=1}^T E[U_s^2|\calF_{s-1}]\right)=&Var\left(\sum_{s=1}^N E[U_s^2|\calF_{s-1}]+\sum_{s=N+1}^T E[U_s^2|\calF_{s-1}]\right)\\
	\lesssim & Var\left(\sum_{s=1}^N E[U_s^2|\calF_{s-1}]\right)+Var\left(\sum_{s=N+1}^T E[U_s^2|\calF_{s-1}]\right)
	\end{align*}
	by the fact that $Cov(X,Y)\le \max\{Var(X),Var(Y)\}$.
	Note also that
	\begin{align*}
	E[U_s^2|\calF_{s-1}]=&\sigma_{W,NT}^{-2}E\left[\left(\sum_{1\le t<s} w_{st} +\sum_{1\le i < s}w_{is}+w_{ss}\right)^2| \calF_{s-1}\right]\\
	=&\sigma_{W,NT}^{-2}\Bigg(\sum_{1\le t<s} \sum_{1\le t'<s}E[w_{st}w_{st'}| \calF_{s-1}] +\sum_{1\le i < s}\sum_{1\le i' < s}E[w_{is}w_{i's}| \calF_{s-1}]+E[w_{ss}^2|\calF_{s-1}]\\
	&+2\sum_{1\le t<s} \sum_{1\le i < s}E[w_{st}w_{is}| \calF_{s-1}] + 2\sum_{1\le t<s} E[w_{st}w_{ss}| \calF_{s-1}]+2\sum_{1\le i < s}E[w_{is}w_{ss}| \calF_{s-1}] 
	\Bigg)\\
	=&
	\sigma_{W,NT}^{-2}\Bigg(\sum_{1\le t<s} \sum_{1\le t'<s}E[w_{st}w_{st'}| \gamma_t,\gamma_{t'}] +\sum_{1\le i < s}\sum_{1\le i' < s}E[w_{is}w_{i's}| \alpha_i,\alpha_{i'}]+E[w_{ss}^2]\\
	&+2\sum_{1\le t<s} \sum_{1\le i < s}E[w_{st}w_{is}| \alpha_i,\gamma_t] 
	+2\sum_{1\le t<s} E[w_{st}w_{ss}| \gamma_t] +2\sum_{1\le i < s}E[w_{is}w_{ss}| \alpha_i] 
	\Bigg)\\
	=&
	\sigma_{W,NT}^{-2}\Bigg(\sum_{1\le t<s} \sum_{1\le t'<s}E[w_{st}w_{st'}| \gamma_t,\gamma_{t'}] +\sum_{1\le i < s}\sum_{1\le i' < s}E[w_{is}w_{i's}| \alpha_i,\alpha_{i'}]+E[w_{ss}^2]
	\Bigg)\\
	=&
	\sigma_{W,NT}^{-2}\Bigg(\sum_{1\le t<s} E[w_{st}^2| \gamma_t]+\sum_{1\le i < s} E[w_{is}^2| \alpha_i]+\sum_{1\le t<s} \sum_{ 1\le t' < s, t'\ne t} E[w_{st}w_{st'}| \gamma_t,\gamma_{t'}]\\
	& +\sum_{1\le i < s}\sum_{1\le i' \le s, i' \ne i} E[w_{is}w_{i's}| \alpha_i,\alpha_{i'}]+E[w_{ss}^2]
	\Bigg),
	\end{align*}
	for $s\le N$, where the second to the last equality follows from the observation that
	\begin{align*}
	E[w_{st}w_{is}| \alpha_i,\gamma_t]=&E[E[w_{st}w_{is}|\alpha_i,\gamma_t,\gamma_s]| \alpha_i,\gamma_t]
	=E[w_{is}E[w_{st}|\gamma_t]| \alpha_i,\gamma_t]=0,
	\\
	E[w_{st}w_{ss}| \gamma_t]=&E[w_{st}E[w_{ss}|\alpha_s,\gamma_t]| \gamma_t]
	=E[w_{st}E[w_{ss}|\alpha_s]| \gamma_t]=0,\\
	E[w_{is}w_{ss}| \alpha_i] 
	=&E[w_{is}E[w_{ss}|\alpha_i,\gamma_s]| \alpha_i] 
	=E[w_{is}E[w_{ss}|\gamma_s]| \alpha_i]=0.
	\end{align*}
	Then, we have
	\begin{align}
	Var\left( \sum_{s=1}^N E[U_s^2|\calF_{s-1}]\right)
	=& 
	\sigma_{W,NT}^{-4}Var\Bigg(\sum_{s=1}^N\sum_{1\le t<s} E[w_{st}^2| \gamma_t]+\sum_{s=1}^N\sum_{1\le i < s} E[w_{is}^2| \alpha_i]\nonumber \\
	&+\sum_{s=1}^N\sum_{1\le t<s} \sum_{1\le t'<s, t'\ne t} E[w_{st}w_{st'}| \gamma_t,\gamma_{t'}]
	+\sum_{s=1}^N\sum_{1\le i < s}\sum_{1\le i'\le s, i' \ne i} E[w_{is}w_{i's}| \alpha_i,\alpha_{i'}]
	\Bigg)\nonumber\\
	=&
	\sigma_{W,NT}^{-4}\Bigg\{Var\Bigg(\sum_{s=1}^N\sum_{1\le t<s} E[w_{st}^2| \gamma_t]
	+\sum_{s=1}^N\sum_{1\le t<s} \sum_{1\le t'<s, t'\ne t}E[w_{st}w_{st'}| \gamma_t,\gamma_{t'}]
	\Bigg)\nonumber\\
	&+
	Var\Bigg(\sum_{s=1}^N\sum_{1\le i < s} E[w_{is}^2| \alpha_i]+\sum_{s=1}^N\sum_{1\le i < s}\sum_{1\le i'<s, i' \ne i} E[w_{is}w_{i's}| \alpha_i,\alpha_{i'}]
	\Bigg)\Bigg\}\nonumber\\
	\lesssim&
	\sigma_{W,NT}^{-4}\Bigg\{Var\Bigg(\sum_{s=1}^N\sum_{1\le t<s} E[w_{st}^2| \gamma_t]\Bigg)
	+Var\Bigg(\sum_{s=1}^N\sum_{1\le t<s} \sum_{ 1\le t'<s, t'\ne t}E[w_{st}w_{st'}| \gamma_t,\gamma_{t'}]
	\Bigg)\nonumber\\
	&+
	Var\Bigg(\sum_{s=1}^N\sum_{1\le i < s} E[w_{is}^2| \alpha_i]\Bigg)+Var\Bigg(\sum_{s=1}^N\sum_{1\le i < s}\sum_{1\le i'<s, i' \ne i} E[w_{is}w_{i's}| \alpha_i,\alpha_{i'}]
	\Bigg)\Bigg\},\label{eq:cond-II_var_1}
	\end{align}
	For the first term on the RHS in (\ref{eq:cond-II_var_1}), we have
	\begin{align*}
	Var\Bigg(\sum_{s=1}^N\sum_{1\le t<s} E[w_{st}^2| \gamma_t]\Bigg)
	=&
	\sum_{s=1}^N\sum_{s'=1}^N\sum_{1\le t<s}\sum_{1\le t'<s}Cov\Bigg( E[w_{st}^2| \gamma_t],E[w_{s't'}^2| \gamma_{t'}]\Bigg)\\
	=&
	\sum_{s=1}^N\sum_{s'=1}^N\sum_{1\le t<s}Cov\Bigg( E[w_{st}^2| \gamma_t],E[w_{s't}^2| \gamma_{t}]\Bigg)\\
	=&
	\sum_{s=1}^N\sum_{1\le t<s}Var\Bigg( E[w_{st}^2| \gamma_t]\Bigg)+
	\sum_{s=1}^N\sum_{1\le s' \le N,s'\ne s}\sum_{1\le t<s}Cov\Bigg( E[w_{st}^2| \gamma_t],E[w_{s't}^2| \gamma_{t}]\Bigg)\\
	\le&
	\sum_{s=1}^N\sum_{1\le t<s} E[w_{st}^4]
	+
	2\sum_{s=1}^N\sum_{1\le s' <s}\sum_{1\le t<s}E[w_{st}^2w_{s't}^2]
	\lesssim 
	N^3 E[w_{11}^4],
	\end{align*}
	where the first inequality follows from Jensen's inequality, the law of total Covariance and the observation that $E\left[Cov(w_{st}^2,w_{s't}^2| \gamma_{t})\right]= 0$ whenever $s\ne s'$. 
	For the second term on the RHS in (\ref{eq:cond-II_var_1}), observe that
	\begin{align}
	&Var\Bigg(\sum_{s=1}^N\sum_{1\le t<s} \sum_{ 1\le t' <s, t'\ne t}E[w_{st}w_{st'}| \gamma_t,\gamma_{t'}]
	\Bigg)\nonumber\\
	=&
	\sum_{s=1}^N \sum_{s'=1}^N\sum_{1\le t<s}\sum_{1\le t'<s'} \sum_{1\le u <s, u\ne t} \sum_{1\le u' <s', u'\ne t'}Cov\Bigg(E[w_{st}w_{su}| \gamma_t,\gamma_{u}],
	E[w_{s't'}w_{s'u'}| \gamma_{t'},\gamma_{u'}]
	\Bigg)\nonumber\\
	=&
	\sum_{s=1}^N \sum_{s'=1}^N\sum_{1\le t<s\wedge s'} \sum_{1\le u< s\wedge s', u\ne t} \sum_{ 1\le u' < s\wedge s', u'\ne u,t}Cov\Bigg(E[w_{st}w_{su}| \gamma_t,\gamma_{u}],
	E[w_{s't}w_{s'u'}| \gamma_{t},\gamma_{u'}]
	\Bigg)\nonumber\\
	&+\sum_{s=1}^N \sum_{s'=1}^N\sum_{1\le t<s\wedge s'} \sum_{1\le u< s\wedge s', u\ne t} Cov\Bigg(E[w_{st}w_{su}| \gamma_t,\gamma_{u}],
	E[w_{s't}w_{s'u}| \gamma_{t},\gamma_{u}]
	\Bigg). \label{eq:cond-II_var_2}
	\end{align}
	In the case with $t=t'$, $u\ne u'$ and $t\ne u$, we have
	\begin{align*}
	E[E[w_{st}w_{su}| \gamma_t,\gamma_{u}]]=&
	E[w_{st}w_{su}]=0
	\end{align*}
	by Lemma \ref{lem:nice} since $u$ appears only once. 
	Hence, the first term on the RHS of (\ref{eq:cond-II_var_2}) becomes
	\begin{align*}
	Cov\Bigg(E[w_{st}w_{su}| \gamma_t,\gamma_{u}],
	E[w_{s't}w_{s'u'}| \gamma_{t},\gamma_{u'}]
	\Bigg)=&
	E\Bigg[E[w_{st}w_{su}| \gamma_t,\gamma_{u}]
	E[w_{s't}w_{s'u'}| \gamma_{t},\gamma_{u'}]
	\Bigg]\\
	=&
	E\Bigg[E\Big[E[w_{st}w_{su}| \gamma_t,\gamma_{u}]
	E[w_{s't}w_{s'u'}| \gamma_{t},\gamma_{u'}]| \gamma_t,\gamma_{u}\Big]
	\Bigg]\\
	=&
	E\Bigg[E[w_{st}w_{su}| \gamma_t,\gamma_{u}]E\Big[
	E[w_{s't}w_{s'u'}| \gamma_{t},\gamma_{u'}]| \gamma_t,\gamma_{u}\Big]
	\Bigg]\\
	=&
	E\Bigg[E[w_{st}w_{su}| \gamma_t,\gamma_{u}]E\Big[
	E[w_{s't}w_{s'u'}| \gamma_{t},\gamma_{u'}]| \gamma_t\Big]
	\Bigg]\\
	=&
	E\Bigg[E[w_{st}w_{su}| \gamma_t,\gamma_{u}]
	E[w_{s't}w_{s'u'} |\gamma_t]
	\Bigg]=0
	\end{align*}
	following the observation that
	\begin{align*}
	E[w_{s't}w_{s'u'} |\gamma_t]=&E[E[w_{s't}w_{s'u'} |\alpha_{s'},\gamma_t]|\gamma_t]]
	=E[w_{s't}E[w_{s'u'} |\alpha_{s'},\gamma_t]|\gamma_t]]
	=E[w_{s't}E[w_{s'u'} |\alpha_{s'}]|\gamma_t]]=0,
	\end{align*}
	where the last equality follows from degeneracy. 
	The RHS of (\ref{eq:cond-II_var_2}) becomes
	\begin{align*}
	\sum_{s=1}^N \sum_{s'=1}^N\sum_{1\le t<s\wedge s'} \sum_{1\le u< s\wedge s', u\ne t} Cov\Bigg(E[w_{st}w_{su}| \gamma_t,\gamma_{u}],
	E[w_{s't}w_{s'u}| \gamma_{t},\gamma_{u}]
	\Bigg)
	\\
	\lesssim N^4 E\Bigg[E[w_{11}w_{12}| \gamma_1,\gamma_{2}]
	E[w_{21}w_{22}| \gamma_{1},\gamma_{2}]
	\Bigg].
	\end{align*}
	Thus, it is of a smaller asymptotic order if 
	\begin{align*}
	\frac{E\big[(E[w_{11}w_{12}| \gamma_1,\gamma_{2}]
		)^2
		\big]}{\{E[w_{11}^2]\}^2}=o(1),
	\end{align*}
	which is guaranteed by Assumption \ref{a:deg_U-stat}.
	We can similarly verify that
	\begin{align*}
	Var\Bigg(\sum_{s=1}^N\sum_{1\le i < s} E[w_{is}^2| \alpha_i]\Bigg)+Var\Bigg(\sum_{s=1}^N\sum_{1\le i < s}\sum_{i' \ne i}^s E[w_{is}w_{i's}| \alpha_i,\alpha_{i'}]
	\Bigg)
	\\
	\lesssim N^3 E[w_{11}^4]+ N^4 E\big[(E[w_{11}w_{21}|\alpha_1,\alpha_2])^2\big]
	\end{align*}
	Thus, Assumption \ref{a:deg_U-stat} implies 
	\begin{align*}
	Var\left( \sum_{s=1}^N E[U_s^2|\calF_{s-1}]\right)
	\lesssim&
	\frac{N^{-1} E[w_{11}^4]+ E\big[(E[w_{11}w_{21}|\alpha_1,\alpha_2])^2\big] \vee E\big[(E[w_{11}w_{12}|\gamma_1,\gamma_2])^2\big]}{\{E[w_{11}^2]\}^2}=o(1).
	\end{align*}
	By analogous arguments, we have
	\begin{align*}
	Var\left( \sum_{s=N+1}^T E[U_s^2|\calF_{s-1}]\right)=o(1).
	\end{align*}
	This concludes verification of Condition II and thus the proof of the lemma.
\end{proof}

\section{Proof of Theorem \ref{thm:CLT_two-way}}\label{sec:thm:CLT_two-way}

	Assume $N\le T$ without loss of generality, and define
	\begin{align*}
	U_1=&\sigma_{LW,NT}^{-1}(a_1+b_1+w_{11}) \quad \calF_0=\sigma(\{\emptyset\}),\\
	U_2=&\sigma_{LW,NT}^{-1}\left(a_2+b_2+w_{21} + w_{22} + w_{12}\right), \quad \calF_1=\sigma( \alpha_1, \gamma_1),\\
	\vdots\\
	U_N=&\sigma_{LW,NT}^{-1}\left(a_N+b_N+\sum_{1\le t<N } w_{Nt} +  w_{NN}+ \sum_{1\le i <N } w_{iN}\right),\quad
	\calF_{N-1}=\sigma( (\alpha_i)_{i=1}^{N-1}, (\gamma_t)_{t=1}^{N-1}),\\
	U_{N+1}=&\sigma_{LW,NT}^{-1}\left(b_{N+1}+\sum_{1\le i\le N } w_{i,N+1}\right),
	\quad \calF_N=\sigma((\alpha_i)_{i=1}^{N}, (\gamma_t)_{t=1}^{N}),\\
	\vdots\\
	U_{T}=&\sigma_{LW,NT}^{-1}\left(b_T+\sum_{1\le i\le N } w_{iT}\right),\quad \calF_{T-1}=\sigma( (\alpha_i)_{i=1}^{N}, (\gamma_t)_{t=1}^{T-1}).
	\end{align*}
	With these notations, note that $\sum_{s=1}^T U_s=L_{NT}+W_{NT}$ and $E[U_s|\calF_{s-1}]=0$ for all $s=1,...,T$.
	In other words, the above sequence forms a martingale difference sequence. 
	To proceed, we verify the following two conditions,
	\begin{align*}
	&\sum_{s=1}^T E[|U_s|^{4}]=o(1) \text{ as } T\to\infty,  &\text{(Condition I)}\\
	&Var\left(\sum_{s=1}^T E[U_s^2|\calF_{s-1}]\right)=o(1) \text{ as } T\to \infty. &\text{(Condition II)}
	\end{align*}
	Once these two conditions are verified, we can then apply the martingale CLT of \cite{heyde1970departure} under their conditions (2) and (4) to conclude the proof of the theorem.
	
	Denote
	\begin{align*}
	U_{1s}=&
	\begin{cases}
	\sigma_{LW,NT}^{-1}(a_s+b_s), \qquad s\le N,\\
	\sigma_{LW,NT}^{-1} b_s,\qquad  N<s\le T,
	\end{cases}\\
	U_{2s}=&
	\begin{cases}
	\sigma_{LW,NT}^{-1}\left(\sum_{1\le i<s } w_{is}+\sum_{1\le t\le s } w_{si}\right), \qquad s\le N,\\
	\sigma_{LW,NT}^{-1} \sum_{N+1\le t\le s } w_{st},\qquad  N<s\le T.
	\end{cases}
	\end{align*}
	Observe that some calculation yields 
	\begin{align}
	\sigma_{LW,NT}^4=N^2(E[a_1^2])^2 + T^2 (E[b_1^2])^2 + N^2T^2 (E[w_{11}^2])^2.\label{eq:sigma_LW_NT^4}
	\end{align}

	Note that the results for the case in which $Var(L_{NT})$ dominates follows straightforwardly from an application of Lyapunov's CLT as $L_{NT}$ is a sum of independent random variables, and the results for the case in which $Var(W_{NT})$ dominates is a direct implication of Lemma \ref{thm:CLT_U-stat}.	
	Therefore, we only show the results for the case in which
	\begin{align*}
	&\lim_{N\to \infty}\frac{Var(L_{NT})}{Var(W_{NT})}\in (0,\infty).
	\end{align*}

	\noindent \textbf{Verifying Condition I.}
	Note that $\sum_{s=1}^TE[|U_s|^4]\lesssim \sum_{s=1}^T (E[U_{1s}^4 ]+E[ U_{2s}^4])$. 
	$\sum_{s=1}^T E[ U_{2s}^4]=o(1)$ follows from Assumption \ref{a:two-way} and (\ref{eq:sigma_LW_NT^4}).
	$ \sum_{s=1}^T E[ U_{2s}^4]  =o(1)$ follows from the same calculations as in the verification of Condition I in the proof of Lemma \ref{thm:CLT_U-stat}.
	
	\bigskip
	
	\noindent \textbf{Verifying  Condition II.} Note that
	\begin{align*}
	E[U_s^2|\calF_{s-1}]=A_{11s}+2A_{12s} + A_{22s},
	\end{align*}
	where $A_{kls}:=E[U_{ks}U_{ls}|\calF_{s-1}]$. 
	Hence
	\begin{align*}
	Var\left(\sum_{s=1}^T E[U_s^2|\calF_{s-1}]\right)
	=&
	Var\left(\sum_{s=1}^T\left(A_{11s} +2 A_{12s} + A_{22s} \right)\right)\\
	\lesssim&
	Var\left(\sum_{s=1}^T A_{11s}\right) +Var\left(\sum_{s=1}^T A_{12s}\right) + Var\left(\sum_{s=1}^TA_{22s} \right)\\
	=&
	\underbrace{Var\left(\sum_{s=1}^T E[U_{1s}^2|\calF_{s-1}]\right)}_{=:(I)}
	+
	\underbrace{Var\left(\sum_{s=1}^T E[U_{1s}U_{2s}|\calF_{s-1}]\right) }_{=:(II)}
	+ 
	\underbrace{Var\left(\sum_{s=1}^T E[U_{2s}^2|\calF_{s-1}]\right)}_{=:(III)}
	\end{align*}
	Under Assumption \ref{a:deg_U-stat}, 
	$(III)=o(1)$ follows from (\ref{eq:sigma_LW_NT^4}) and the same calculations as in the verification of Condition II in the proof of Lemma \ref{thm:CLT_U-stat} with $U_s$ replaced by $U_{2s}$. 
	For the $(I)$ term, observe that
	\begin{align*}
	(I)=\sigma_{LW,NT}^{-4}E\left[\left(\sum_{s=1}^N a_s^2+\sum_{s=1}^T b_s^2\right)^2 \right]
	=\sigma_{LW,NT}^{-4} \left(N E[a_1^4] + TE[b_1^4]\right)=o(1),
	\end{align*}
	where the last equality is guaranteed under Assumption \ref{a:two-way} and (\ref{eq:sigma_LW_NT^4}).
	
	For the $(II)$ term, we have
	\begin{align*}
	(II)=&Var\left(\sum_{s=1}^T E[U_{1s}U_{2s}|\calF_{s-1}]\right) \\
	=&
	\sigma_{LW,NT}^{-4}Var\left(\sum_{s=1}^N E\left[(a_s+b_s)\left(\sum_{1\le i<s } w_{is}+\sum_{1\le t\le s } w_{si}\right)|\calF_{s-1}\right]+
	\sum_{s=N+1}^TE\left[b_s\left(\sum_{1\le t\le s }  w_{st}\right) |\calF_{s-1}\right]\right)\\
	\lesssim&
	\sigma_{LW,NT}^{-4}\left\{Var\left(\sum_{s=1}^N E\left[(a_s+b_s)\left(\sum_{1\le i<s } w_{is} 
	+
	\sum_{1\le t\le s } w_{si}\right)|\calF_{s-1}\right]\right)
	\right.\\
	&\qquad\qquad \left. 
	+Var\left(
	\sum_{s=N+1}^TE\left[b_s\left(\sum_{1\le t\le s }  w_{st}\right) |\calF_{s-1}\right]\right)\right\}.
	\end{align*}
	For the first term in the last expression, observe that
	\begin{align*}
	&Var\left(\sum_{s=1}^N E\left[(a_s+b_s)\left(\sum_{1\le i<s } w_{is}+\sum_{1\le t\le s } w_{si}\right)|\calF_{s-1}\right]\right)\\
	\lesssim&
	Var\left(\sum_{s=1}^N \sum_{1\le i<s }E\left[ a_iw_{is}|\calF_{s-1}\right]\right)+Var\left(\sum_{s=1}^N \sum_{1\le t\le s}E\left[ b_t w_{si}|\calF_{s-1}\right]\right)=o(\sigma_{LW,NT}^4),
	\end{align*}
	follows from Assumption 2 and (\ref{eq:sigma_LW_NT^4}), since the condition
	\begin{align*}
	\frac{N^3E\left[(a_1  w_{12})^2\right] + T^3E\left[(b_1 w_{21})^2\right]}{\sigma_{LW,NT}^4}=o(1) 
	\end{align*} 
	implies
	\begin{align*}
	E\left[\sum_{s=1}^N  \sum_{i=1}^sE[a_s  w_{si}|\calF_{s-1}]\right]^2 + E\left[\sum_{s=1}^T \sum_{i=1}^s E[b_s w_{is}|\calF_{s-1}]\right]^2=o(\sigma_{LW,NT}^4) .
	\end{align*}
	A similar argument shows
	\begin{align*}
	Var\left(
	\sum_{s=N+1}^TE\left[b_s\left(\sum_{1\le t\le s }  w_{st}\right) |\calF_{s-1}\right]\right)=o(\sigma_{LW,NT}^4).
	\end{align*}
	This verifies Condition II and thus concludes the proof. 
\qed

\bibliography{biblio}
\end{document}